\documentclass[11pt]{article}
\usepackage{amsmath}
\usepackage{amssymb}
\usepackage{amsthm}
\usepackage{graphicx}
\usepackage{mathtools}
\usepackage[margin=1.1in]{geometry}
\usepackage{hyperref}

\renewcommand{\le}{\leqslant}

\renewcommand{\ge}{\geqslant}

\newcommand{\F}{\mathbb{F}}
\newcommand{\I}{\mathcal{I}}

\newtheorem{theorem}{Theorem}
\newtheorem{corollary}[theorem]{Corollary}
\newtheorem{prop}[theorem]{Proposition}
\newtheorem{lemma}[theorem]{Lemma}
\newtheorem{defn}{Definition}

\begin{document}

\title{An Exponential Lower Bound on the Sub-Packetization of Minimum Storage Regenerating Codes\thanks{An earlier version this work was presented at the 2019 ACM Symposium on Theory of Computing (STOC)~\cite{AG-stoc19}. The current version includes a slightly improved lower bound.}}
\author{Omar Alrabiah\thanks{Department of Mathematical Sciences, Carnegie Mellon University, Pittsburgh, PA. Email: {\tt oalrabia@andrew.cmu.edu}} \and Venkatesan Guruswami\thanks{Computer Science Department, Carnegie Mellon University, Pittsburgh, PA 15213. Email: {\tt venkatg@cs.cmu.edu}. Research supported in part by NSF CCF-1563742.}}
	
\date{}
\maketitle
\thispagestyle{empty}

\begin{abstract}
An $(n,k,\ell)$-vector MDS code over a field $\F$ is a $\F$-linear subspace of $(\F^\ell)^n$  of dimension $k\ell$, such that any $k$ (vector) symbols of the codeword suffice to determine the remaining $r=n-k$ (vector) symbols. The length $\ell$ of each codeword symbol is called the \emph{sub-packetization} of the code. Such a code is called minimum storage regenerating (MSR), if any single symbol of a codeword can be recovered by downloading $\ell/r$ field elements (which is known to be the minimum possible) from each of the other symbols. 
	
\smallskip MSR codes are attractive for use in  distributed storage systems, and by now a variety of ingenious constructions of MSR codes are available. However, they all suffer from exponentially large sub-packetization $\ell \gtrsim r^{k/r}$. Our main result is an almost tight lower bound showing that for an MSR code, one must have $\ell \ge \exp(\Omega(k/r))$. Previously, a lower bound of $\approx \exp(\sqrt{k/r})$, and a tight lower bound for a restricted class of "optimal access" MSR codes, were known. 
\end{abstract}

\section{Introduction}

Traditional Maximum Distance Separable (MDS) codes such as Reed-Solomon codes provide the optimal trade-off between redundancy and number of worst-case erasures tolerated. When encoding $k$ symbols of data into an $n$ symbol codeword by an $(n,k)$-MDS code, the data can be recovered from any set of $k$ out of $n$ codeword symbols, which is clearly the best possible. MDS codes are thus a  a naturally appealing choice to minimize storage overhead in distributed storage systems (DSS). One can encode data, broken into $k$ pieces, by an $(n,k)$-MDS code, and distribute the $n$ codeword symbols on $n$ different storage nodes, each holding the symbol corresponding to one codeword position. In the sequel, we use the terms storage node and codeword symbol interchangeably. 

A rather common scenario faced by  modern large scale DSS is the failure or temporary unavailability of storage nodes. It is of great importance to promptly respond to such failures, by efficient repair/regeneration of the failed node using the content stored in some of other nodes (which are called ``helper" nodes as they assist in the repair). This requirement has spurred a set of fundamentally new and exciting challenges concerning codes for recovery from erasures, with the goal of balancing worst-case fault tolerance from many erasures, with very efficient schemes to recover from the much more common scenario of single (or a few) erasures.

There are two measures of repair efficiency that have received a significant amount of attention in the last decade. One concerns \emph{locality}, where we would like to repair a node \emph{locally} based on the contents of a small number of other storage nodes. Such locality necessarily compromises the MDS property, and a rich body of work on locally repairable codes (LRCs) studies the best trade-offs possible in this model and constructions achieving those~\cite{GHSY12,PD14,TB14}. The other line of work, which is the subject of this paper, focuses on optimizing the amount of data downloaded from the other nodes. This model allows the helper node to respond with a fraction of its contents. The efficiency measure is the \emph{repair bandwidth}, which is the total amount of data downloaded from all the helper nodes. Codes in this model are called regenerating codes, and were systematically introduced in the seminal work of Dimakis et al.~\cite{DGWWR10}, and have since witnessed an explosive amount of research.

Rather surprisingly, even for some MDS codes, by contacting more helper nodes but downloading fewer symbols from each, one can do much better than the ``usual" scheme, which would download the contents of $k$ nodes in full. In general an entire spectrum of trade-offs is possible between storage overhead and repair bandwidth. This includes \emph{minimum bandwidth regenerating (MBR) codes} with the minimum repair bandwidth of $\ell$~\cite{RSK11}.   At the other end of the spectrum, we have \emph{minimum storage regenerating (MSR) codes} defined formally below) which retain the MDS property and thus have optimal redundancy. This work focuses on MSR codes.

\smallskip\noindent \textbf{Example.}
We quickly recap the classic example of the EVENODD code~\cite{BBBM94,DRWS-survey} to illustate regeneration of a lost symbol in an MDS code with non-trivial bandwidth. This is an $(4,2)$ MDS code with $4$ storage nodes, each storing a vector of two symbols over the binary field. We denote by $\mathcal{P}_1,\mathcal{P}_2$ the two parity nodes.
\begin{center}
\begin{tabular}{| c | c | c | c |} 
\hline
$\mathcal{S}_1$ & $\mathcal{S}_2$ & $\mathcal{P}_1$ & $\mathcal{P}_2$ \\
\hline
\; $a_1$ \; & \; $b_1$ \; & \; $a_1 + b_1$ \; & \; $a_2 + b_1$ \; \\
\hline
\; $a_2$ \; & \; $b_2$ \; & \; $a_2 + b_2$ \; & \; $a_1 + a_2 + b_2$ \; \\
\hline
\end{tabular}
\end{center}
The naive scheme to repair a node would contact any two of the remaining three nodes, and download both bits from each of them, for a total repair bandwidth of $4$ bits. However, it turns out that one can get away with downloading just one bit from each of the three other nodes, for a repair bandwidth of $3$ bits! If we were to repair the node $\mathcal{S}_1$, the remaining nodes $(\mathcal{S}_2, \mathcal{P}_1, \mathcal{P}_2)$ would send $(b_1, a_1+b_1, a_2+b_1)$, respectively. If we were to repair the node $\mathcal{S}_2$, the remaining nodes $(\mathcal{S}_1, \mathcal{P}_1, \mathcal{P}_2)$ would send $(a_2, a_2+b_2, a_2+b_1)$, respectively. If we were to repair the node $\mathcal{P}_1$, the remaining nodes $(\mathcal{S}_1, \mathcal{S}_2, \mathcal{P}_2)$ would send $(a_1, b_1, a_1+a_2+b_2)$, respectively. If we were to repair the node $\mathcal{P}_2$, the remaining nodes $(\mathcal{S}_1, \mathcal{S}_2, \mathcal{P}_1)$ would send $(a_2, b_1, (a_1+b_1) + (a_2+b_2))$, respectively.  Note that in the last case, the helper node $\mathcal{P}_1$ sends a linear combination of its symbols---this is in general a powerful ability that we allow in MSR codes.

\smallskip\noindent \textbf{Vector codes and sub-packetization.}
The above example shows that when the code is an $(n,k)$ \emph{vector} MDS code, where each codeword symbol itself is a \emph{vector}, say in $\F^\ell$ for some field $\F$, then one can hope to achieve repair bandwidth smaller than then naive $k\ell$. The length of the vector $\ell$ stored at each node is called the \emph{sub-packetization} (since this is the granularity to which a single codeword symbol needs to be divided into). 

\smallskip
\noindent \textbf{MSR codes.}
A natural question is how small a repair bandwidth one can achieve with MDS codes. The so-called \emph{cutset bound}~\cite{DGWWR10} dictates that one must download at least $(n-1)\ell/(n-k)$ symbols of $\F$ from the remaining nodes to recover any single node. Further, in order to attain this optimal repair bandwidth bound, each of the $(n-1)$ nodes must respond with $\ell/(n-k)$ field elements. Vector MDS codes which admit repair schemes meeting the cutset bound (for repair of every node) are called minimum storage regenerating (MSR) codes (for the formal description, see Definition~\ref{def:msr-code}). MSR codes, and specifically their sub-packetization, are the focus of this paper. 

\smallskip
\noindent \textbf{Large sub-packetization: problematic and inherent.}
While there are many constructions of MSR codes by now, they all have large sub-packetization, which is at least $r^{k/r}$. For the setting of most interest, when we incur a small redundancy $r$ in exchange for repair of information, this is very large, and in particular $\exp(\Omega(k))$ when $r=O(1)$. A small sub-packetization is important for a number of reasons, as explained in some detail in the introduction of \cite{RTGE18}.  A large sub-packetization limits the number of storage nodes (for example if $\ell \ge \exp(\Omega(n))$, then $n=O(\log \ell)$ where $\ell$ is the storage capacity of each node), and in general leads to a reduced design space in terms of various systems parameters. A larger sub-packetization also makes management of meta-data, such as description of the code and the repair mechanisms for different nodes, more difficult. For a given storage capacity, a smaller sub-packetization allows one to distribute codewords corresponding to independently coded files among multiple nodes, which allows for distributing the load of providing information for the repair of a failed node among a larger number of nodes.

It has been known that somewhat large sub-packetization is inherent for MSR codes (we will describe the relevant prior results in the next section). In this work, we improve this lower bound to exponential, showing that unfortunately the exponential sub-packetization of known constructions is inherent. Our main result is the following.

\begin{theorem}
\label{thm:main-intro}
Suppose an $(n,k)$-vector MDS code with redundancy $r=n-k \ge 2$ is minimum storage regenerating (MSR). Then its sub-packetization $\ell$ must satisfy\footnote{In the conference version~\cite{AG-stoc19}, a weaker lower bound of  $e^{(k-1)/(4r)}$ was shown.}
\[ \ell \ge \left(\frac{r^2}{r^2-r+1}\right)^{(k-1)/2} \ge e^{(k-1)(r-1)/(2r^2)} \ . \]
\end{theorem}

Our lower bound almost matches the sub-packetization of $r^{O(k/r)}$ achieved by the best known constructions. Improving the base of the exponent in our lower bound to $r$ will make it even closer to the upper bounds. Though when $r$ is small, which is the primary setting of interest in codes for distributed storage, this difference is not that substantial. We remark that our theorem leaves out the case when $r=1$, which is known to have a sub-packetization of $\ell=1$~\cite{GoparajuFV16}.

\smallskip\noindent \textbf{A few words about our proof.} Previous work~\cite{TWB14} has shown that an $(n,k)$ MSR code with sub-packetization $\ell$ implies a family of $(k-1)$ $\ell/r$-dimensional subspaces $H_i$ of $\F^\ell$ each of which has an associated collection of $(r-1)$ linear maps obeying some strong properties. For instance, in the case $r=2$, there is an invertible map $\phi_i$ associated with $H_i$ for each $i$ which leaves all subspaces $H_j$, $j \neq i$, invariant, and maps $H_i$ itself to a disjoint space (i.e., $\phi_i(H_i) \cap H_i = \{0\}$). 
The task of showing a lower bound on $\ell$ then reduces to the linear-algebraic challenge of showing an upper bound on the size of such a family of subspaces and linear transformations, which we call an MSR subspace family (Definition~\ref{def:msr-family}). 
The authors of \cite{GTC14} showed an upper bound $O(r \log^2 \ell)$ on the size of MSR subspace families via a nifty partitioning and linear independence argument. 

We follow a different approach by showing that the number of linear maps that fix all subspaces in an MSR family decreases sharply as the number of subspaces increases. Specifically, we show that dimension of the linear space of such linear maps decreases exponentially in the number of subspaces in the MSR family. This enables us to prove an $O(r \log \ell)$ upper bound. This bound is asymptotically tight (up to a $O(\log r)$ factor), as there is a construction of an MSR subspace family of size $(r+1)\log_r \ell$~\cite{WTB12}. We also present an alternate construction in Section~\ref{sec:construction}, which works for all fields with more than $2$ elements, compared to the large field size (of at least $\approx r^r \ell$) required in \cite{WTB12}.

\medskip
We now proceed to situate our work in the context of prior work, both constructions and lower bounds, for MSR codes.

\section{Prior and Related Work}

The literature on regenerating codes, and even just MSR codes, is vast with numerous models and constraints, and many incomparable results. Here we only mention the ones closely related to our work and its context --- MSR codes for exact repair with $n-1$ helper nodes, focusing primarily on their sub-packetization.

\smallskip
\noindent \textbf{MSR code constructions.}
We begin code constructions/existence results. Rashmi et al. present an explicit construction of MSR codes with small sub-packetization $\ell \le r$ when the code rate $k/n$ is at most $1/2$~\cite{RSK11}. Cadambe et al.~\cite{CJMRS13} show the existence of high rate MSR codes when the sub-packetization approaches infinity.  Motivated by this result, the problem of designing high-rate MSR codes with finite sub-packetization level is explored in \cite{PapDimCad13, zigzag13, SAK15, WTB12, Cadambe_poly, GoparajuFV16, YeB17a, YeB17b, SVK16} and references therein. In particular, Sasidharan et al.~\cite{SAK15} show the existence of MSR codes with the sub-packetization level  $\ell = r^{\lceil \frac{n}{r} \rceil}$. Such a result with similar sub-packetization levels for repair of only $k$ systematic nodes was obtained earlier in \cite{WTB12, Cadambe_poly}. In order to ensure the MDS property, these results relied on huge fields and randomized construction of the parity check matrices. 

In two fascinating (independent) works, Ye and Barg~\cite{YeB17b} and Sasidharan, Vajha, and Kumar~\cite{SVK16} give a fully explicit construction of MSR codes over small fields with sub-packetization level $\ell = r^{\lceil \frac{n}{r}\rceil}$. These constructions also have the so-called \emph{optimal-access} or \emph{help-by-transfer} property, which means that the helper nodes do not have to perform any linear combinations on their data, and can simply transfer a suitable subset of $\ell/r$ coordinates of the vector in $\F^\ell$ that they store. Thus the number of symbols \emph{accessed} at a node equals the number of symbols it transmits over the network to aid the repair (recall that the repair-bandwidth measures the latter amount).

\smallskip
\noindent \textbf{Sub-packetization lower bounds.}
In summary, while there are several constructions of high rate MSR codes, they all incur large sub-packetization, which is undesirable as briefly explained earlier. This has been partially explained by lower bounds on $\ell$ in a few previous works. For the special case of optimal-access MSR codes, a lower bound of $\ell \ge r^{k/r}$ was shown in \cite{TWB14}, and this was improved (when all-node repair is desired) to $\ell \ge r^{n/r}$ recently~\cite{BK18}. Together with the above-mentioned constructions, we thus have matching upper and lower bounds on $\ell$ for the optimal-access case. This help-by-transfer setting is primarily combinatorial in nature, which is exploited heavily in these lower bounds.

However, lower bounds for general MSR codes, that allow helper nodes to transmit linear combinations of their comments, are harder to obtain. Such a lower bound must rule out a much broader range of possible repair schemes, and must work in an inherently linear-algebraic rather than combinatorial setting. Note that the simple example presented above also used linear combinations in repairing one of the nodes. An MSR code construction with sub-packetization $\ell \le r^{k/(r+1)}$, which beats the above lower bound for optimal-access codes and thus shows a separation between these models, was given in \cite{WTB12}.

Turning to known lower bounds on $\ell$, a weak bound of $\ell {{\ell}\choose{\ell/r}} \ge k$ was shown via a combinatorial argument in \cite{TWB14}. Using an elegant linear independence and partitioning argument, the following bound is proven in \cite{GTC14}:
\begin{equation}
\label{eq:gtc}
2 \log_2(\ell) (\log_{r/(r-1)}(\ell) + 1)  \ge k-1 \ .
\end{equation}
(This was slightly improved in \cite{HPX18}, but the improvement is tiny for the case when $k > r$ which is our focus.) The bound \eqref{eq:gtc} implies a lower bound on sub-packetization of $\ell \ge 2^{\Omega(\sqrt{k/r})}$. Even for the case $r=2$, it was not known if one can achieve sub-packetization smaller than $2^{\Omega(k)}$. Our Theorem~\ref{thm:main-intro} now rules out this possibility. We conjecture that our bound can be improved to $k \le (r+1) \log_r \ell + O(1)$ which will show that the construction in \cite{WTB12} is exactly tight.

\smallskip
\noindent \textbf{Variants of MSR codes.}
While most constructions of MSR and regenerating codes were tailormade vector codes, it was shown in \cite{GW16} that the classical family of Reed-Solomon (RS) themselves can allow for non-trivially bandwidth-efficient repair. This later led to the carefully constructed RS codes which can be repaired with the optimal bandwidth meeting the cutset bound~\cite{TYB-focs17} --- in other words, certain RS codes themselves are MSR codes! However, these RS codes have even larger sub-packetization of $2^{O(n\log n)}$ and this was also shown to be necessary in the form of a lower bound of $2^{\Omega(k\log k)}$ in \cite{TYB-focs17}. 

A slight relaxation MSR codes called $\epsilon$-MSR codes where the helper nodes are allowed to transmit a factor $(1+\epsilon)$ more than the cutset bound, i.e., $(1+\epsilon)\ell/r$ symbols, were put forth in \cite{RTGE18}. They showed that one can construct $\epsilon$-MSR codes with sub-packetization $r^{O(r/\epsilon)} \log n$, and roughly logarithmic sub-packetization is also necessary.

Regenerating and MSR codes have close connections to communication-efficient secret sharing schemes, which were studied and developed in \cite{HLKB15}. In this context, the sub-packetization corresponds to the size of the shares that the parties must hold.

\section{Preliminaries}
\label{sec:prelims}
We will now define MSR codes more formally. We begin by defining vector codes. Let $\F$ be a field, and $n,\ell$ be positive integers. For a positive integer $b$, we denote $[b]=\{1,2,\dots,b\}$. A vector code $C$ of block length $n$ and \emph{sub-packetization} $\ell$ is an $\F$-linear subspace of $(\F^{\ell})^n$. We can express a codeword of $C$ as $\mathbf{c} = (\mathbf{c_1}, \mathbf{c_2}, \dots,\mathbf{c_n})$, where for $i \in [n]$, the block $\mathbf{c_i} = (c_{i,1},\dots,c_{i,\ell})\in \F^\ell$ denotes the length $\ell$ vector corresponding to the $i$'th \emph{code symbol} $\mathbf{c_i}$.

Let $k$ be an integer, with $1 \le k \le n$. If the dimension of $C$, as an $\F$-vector space, is $k \ell$, we say that $C$ is an $(n,k,\ell)_\F$-vector code. The codewords of an $(n,k,\ell)_\F$-vector code are in one-to-one correspondence with vectors in $(\F^\ell)^k$, consisting of $k$ blocks of $\ell$ field elements each. 

Such a code is said to be Maximum Distance Separable (MDS), and called an $(n,k,\ell)$-MDS code (over the field $\F$), if every subset of $k$ code symbols $\mathbf{c_{i_1}}, \mathbf{c_{i_2}}, \dots, \mathbf{c_{i_k}}$ is an information set for the code, i.e., knowing these symbols determines the remaining $n-k$ code symbols and thus the full codeword. An MDS code thus offers the optimal erasure correction property --- the information can be recovered from any set of $k$ code symbols, thus tolerating the maximum possible number $n-k$ of worst-case erasures. 

An $(n,k,\ell)$-MDS code can be used in distributed storage systems as follows. Data viewed as $k\ell$ symbols over $\F$ is encoded using the code resulting in $n$ vectors in $\F^\ell$, which are stored in $n$ storage nodes. Downloading the full contents from any subset of these $k$ nodes (a total of $k\ell$ symbols from $\F$) suffices to reconstruct the original data in entirety. Motivated by the challenge of efficient regeneration of a failed storage node, which is a fairly typical occurrence in large scale distributed storage systems, the repair problem aims to recover any single code symbol $\mathbf{c_i}$ by downloading fewer than $k \ell$ field elements. This is impossible if one only downloads contents from $k$ nodes, but becomes feasible if one is allowed  to contact $h > k$ \emph{helper} nodes and receive fewer than $\ell$ field elements from each.

Here we focus our attention to only repairing the first $k$ code symbols, which we view as the information symbols. This is called "systematic node repair" as opposed to the more general "all node repair" where the goal is to repair all $n$ codeword symbols. We will also only consider the case $h=n-1$, when all the remaining nodes are available as helper nodes. Since our focus is on a lower bound on the sub-packetization $\ell$, this only makes our result stronger, and keeps the description somewhat simpler. We note that the currently best known constructions allow for all-node repair with optimal bandwidth from any subset of $h$ helper nodes.

Suppose we want to repair the $m$'th code symbol for some $m \in [k]$. We download from the $i$'th code symbol, $i \neq m$, a function $h_{i,m}(\mathbf{c_i})$ of its contents, where $h_{i,m} : \F^\ell \to \F^{\beta_{i,m}}$ is the repair function. If we consider the linear nature of $C$, then we should expect from $h_{i,m}$ to utilize it. Therefore, throughout this paper, we shall assume linear repair of the failed node. That is, $h_{i,m}$ is an $\F$-linear function. Thus, we download from each node certain linear combinations of the $\ell$ symbols stored at that node. The total \emph{repair bandwidth} to recover $\mathbf{c_m}$ is defined to be $\sum_{i \neq m} \beta_{i,m}$. By the cutset bound for repair of MDS codes~\cite{DGWWR10}, this quantity is lower bounded by $(n-1) \ell/r$, where $r=n-k$ is the redundancy of the code. Further, equality can be attained only if $\beta_{i,m} = \ell/r$ for all $i$. That is, we download $\ell/r$ field elements from each of the remaining nodes. MDS codes achieving such an optimal repair bandwidth are called \emph{Minimum Storage Regenerating (MSR) codes}, as precisely defined below.

\begin{defn}[MSR code]
\label{def:msr-code}
Let $1 \le k \le n$ and $\ell$ be integers with $r = n-k$ dividing $\ell$. An $(n,k,\ell)$-MDS code $C$ over a field $\F$ is said to be an \emph{$(n,k,\ell)$-MSR code} if for every $m \in [k]$, there are linear functions $h_{i,m} : \F^\ell \to \F^{\ell/r}$, $i \in [n] \setminus \{m\}$, such that the code symbol $\mathbf{c_m}$ of a codeword $\mathbf{c} \in C$ can be computed by an $\F$-linear operation on $\langle h_{i,m}(\mathbf{c_i}) \mid i \in [n] \setminus \{m\}\rangle \in \F^{(n-1)\ell/r}$. 
\end{defn}

\section{Linear-algebraic setup of repair of MSR code}
\label{sec:repair-setup}
In this section, we will setup the repair problem for MSR codes more precisely, leading to a purely linear-algebraic formulation in terms of a collection of subspaces and associated invertible maps. We follow the explanation presented in \cite{TWB14}. 

Let $\mathcal{C}\subseteq (\F^\ell)^n$ be an $(n,k,\ell)$-MSR code, with redundancy $r=n-k$. The MDS property implies that any subset of $k$ codeword symbols determine the whole codeword. We view the first $k$ symbols as the "systematic" ones, with $r$ parity check symbols computed from them, where we remind that when we say code symbol we mean a vector in $\F^\ell$.  So we can assume that there are invertible matrices 
$C_{i,j} \in \F^{\ell \times \ell}$ for $i\in [r]$ and $j \in [k]$ such that for $\mathbf{c} = (\mathbf{c_1}, \mathbf{c_2}, \ldots , \mathbf{c_n}) \in \mathcal{C}$, we have
	\begin{equation*}
	\mathbf{c_{k+i}} = \sum_{j=1}^k{C_{i, j}\mathbf{c_j}} \ .
	\end{equation*}
Suppose we want to repair a systematic node $\mathbf{c_m}$ for $m \in [k]$ with optimal repair bandwidth, by receiving from each of the remaining $n-1$ nodes, $\ell/r$ $\F$-linear combinations of the information they stored. This means that there are repair matrices $S_{1, m}, \ldots , S_{r, m} \in \F^{\ell/r \times \ell}$, such that parity node $k+i$ sends the linear combination 
\begin{equation}
\label{eq:lin-comb-repair}
S_{i, m}\mathbf{c_{k+i}} = S_{i, m}\sum_{j=1}^k{C_{i, j}\mathbf{c_j}}
\end{equation}

Therefore, the information about $\mathbf{c_m}$ that is sent to it by $\mathbf{c_{k+i}}$ is $S_{i,m}C_{i, m}\mathbf{c_m}$. Since the $k$ systematic nodes are independent of each other, then the only way to recover $\mathbf{c_m}$ is by taking a linear combination of $S_{i,m}C_{i, m}\mathbf{c_m}$ for $i \in [r]$ such that the linear combination equals $\mathbf{c_m}$ for any $\mathbf{c_m} \in \F^\ell$. Therefore, to ensure full regeneration of $\mathbf{c_m}$, we must satisfy
\begin{equation*}
\text{rank}\begin{bmatrix}
S_{1, m}C_{1, m} \\ 
S_{2, m}C_{2, m} \\ 
\vdots \\
S_{r, m}C_{r, m} \\ 
\end{bmatrix}
= \ell
\end{equation*}
Since each $S_{i,m} C_{i,m}$ has $\ell/r$ rows, the above happens if and only if 
\begin{equation}
\label{eq:self-disjoint}
\bigoplus_{i=1}^r \mathcal{R}(S_{i,m} C_{i,m}) = \F^\ell
\end{equation}
 where $\mathcal{R}(M)$ denotes the row-span of a matrix $M$.

\subsection{Cancelling interference of other systematic symbols}

Now, for every other systematic node $m' \in [k] \setminus \{m\}$, the parity nodes send the following information linear combinations of $\mathbf{c_{m'}}$
\begin{equation}
\label{eq:cancel}
\begin{bmatrix}
S_{1, m}C_{1, m'} \\ 
S_{2, m}C_{2, m'} \\ 
\vdots \\
S_{r, m}C_{r, m'} \\ 
\end{bmatrix} \mathbf{c_{m'}}
\end{equation}
In order to cancel this from the linear combinations \eqref{eq:lin-comb-repair} received from the parity nodes, the systematic node $m'$ has to send the linear combinations \eqref{eq:cancel} about its contents. To achieve optimal repair bandwidth of at most $\ell/r$ symbols from every node, this imposes the requirement
\begin{equation*}
\text{rank}\begin{bmatrix}
S_{1, m}C_{1, m'} \\ 
S_{2, m}C_{2, m'} \\ 
\vdots \\
S_{r, m}C_{r, m'} \\ 
\end{bmatrix} \le \frac{\ell}{r}
\end{equation*}
However since $C_{i,m'}$ is invertible, and $S_{i,m}$ has full row rank, $\text{rank}(S_{i, m}C_{i, m'}) = \ell/r$ for all $i \in [r]$. Combining this fact with the rank inequality above, this implies
\begin{equation}
\label{eq:fixed-by-others}
\mathcal{R}(S_{1, m}C_{1, m'}) = \cdots = \mathcal{R}(S_{r, m}C_{r, m'})
\end{equation}
for every $m \neq m' \in [k]$, where $\mathcal{R}(M)$ is the row-span of a matrix $M$. 

\subsection{Constant repair matrices and casting the problem in terms of subspaces}

We now make an important simplification, which allows us to assume that the matrices $S_{i,m}$ above depend only on the node $m$ being repaired, but not on the helping parity node $i$. That is, $S_m = S_{i,m}$ for all $i \in [r]$. We call repair with this restriction as possessing \emph{constant repair matrices}. It turns out that one can impose this restriction with essentially no loss in parameters --- by Theorem~2 of \cite{TWB14}, if there is a $(n,k,\ell)$-MSR code then there is also a $(n-1,k-1,\ell)$-MSR code with constant repair matrices. 

This allows us to cast the requirements \eqref{eq:self-disjoint} and \eqref{eq:fixed-by-others} in terms of a nice property about subspaces and associated invertible maps, which we abstract below. This property was shown to be intimately tied to MSR codes in \cite{WTB12,TWB14}.

\begin{defn}[MSR subspace family]
\label{def:msr-family}
For integers $\ell,r$ with $r | \ell$ and a field $\F$, a collection of subspaces $H_1,\dots,H_k$ of $\F^\ell$ of dimension $\ell/r$ each is said to be an \emph{$(\ell,r)_\F$-MSR subspace family} if there exist invertible linear maps $\Phi_{i,j}$ on $\F^\ell$, $i \in \{1,2,\dots,k\}$ and $j \in \{1,2,\dots,r-1\}$ such that for every $i \in [k]$, the following holds:
\begin{align}
H_i &\oplus \bigoplus_{j=1}^{r-1}{\Phi_{i, j}(H_i)} = \F^\ell \label{eq:msr-1} \\ 
\Phi_{i', j}(H_i) = H_i &\text{ for every $j \in [r-1]$, and $i' \neq i$} \label{eq:msr-2}
\end{align}
\end{defn}

Now, we recall the argument that if we have an $(n,k,\ell)$-MSR code with constant repair matrices, then that also yields a family of subspaces and maps with the above properties. Indeed, we can take $H_m$, $m \in [k]$, to be $\mathcal{R}(S_m)$, and $\Phi_{m,j}$, $j \in [r-1]$, is the invertible linear transformation mapping $\mathbf{x} \in \F^\ell$, viewed as a row vector, to $\mathbf{x} C_{j+1,m} C_{1,m}^{-1}$. It is clear that Property \eqref{eq:msr-1} follows from \eqref{eq:self-disjoint}, and Property \eqref{eq:msr-2} follows from \eqref{eq:fixed-by-others}. Together with the loss of one dimension in the transformation~\cite{TWB14} to an MSR code with constant repair subspaces, we can conclude the following connection between MSR codes and the very structured set of subspaces and maps of Definition~\ref{def:msr-family}.
\begin{prop}
	\label{prop:connection}
	Suppose there exists an $(n,k,\ell)$-MSR code over a field $\F$. Then there exists an $(\ell,r)_\F$-MSR subspace family with $k-1$ subspaces.
\end{prop}

For the reverse direction, the MSR subspace family can take care of the node repair, but one still needs to ensure the MDS property. This approach was taken in \cite{WTB12}, based on a construction of an $(\ell,r)_\F$-MSR subspace family of size $(r+1)\log_r \ell$. For completeness, we present another construction of an MSR subspace family in Section~\ref{sec:construction}. The subspaces in our construction are identical to \cite{WTB12} but we pick the linear maps differently, using just two distinct eigenvalues. As a result, our construction works over any field with more than two elements. In comparison, the approach in \cite{WTB12} used $k^{r-1} \ell/r$ distinct eigenvalues, and thus required a field that is bigger than this bound. It is an interesting question to see if the MDS property can be incorporated into our construction to give MSR codes with sub-packetization $r^{k/(r+1)}$ over smaller fields.

\section{Limitation of MSR subspace families}
In this section, we state and prove the following strong upper bound on the size of an MSR family of subspaces, showing that the construction claimed in Theorem~\ref{theorem:construction} is not too far from the best possible. This upper bound together with Proposition~\ref{prop:connection} immediately implies our main result, Theorem~\ref{thm:main-intro}.

\begin{theorem}
	\label{thm:subspaces-upperbound}
An $(\ell,r)_\F$-MSR subspace family can have at most 
\begin{equation*}
\frac{2\ln{\ell}}{\ln\left(\frac{r^2}{r^2-r+1}\right)}
\le \left(\tfrac{2r^2}{r-1}\right)\ln{\ell}
\end{equation*}
subspaces.
\end{theorem}

In the rest of the section, we prove the above theorem. Let $H_1,H_2,\dots,H_k$ be the subspaces in an $(\ell,r)_\F$-MSR subspace family with associated invertible linear maps $\Phi_{i, j}$ where $i \in [k]$ and $j \in [r-1]$. Note that these linear maps are in some sense statements about the structure of the spaces $H_1, H_2, \ldots, H_k$. They dictate the way the subspaces can interact with each other, thereby giving rigidity to the way they are structured.
 
The major insight and crux of the proof is the following definition on collections of subspaces. This definition is somewhat inspired by Galois Theory, in that we are looking at the space of linear maps on the vector space $\F^\ell$ that fix all the subspaces in question.

\begin{defn}
\label{def:invariant}
In the vector space $\mathcal{L}(\mathbb{F}^\ell, \mathbb{F}^\ell)$ of all linear maps from $\F^\ell$ to $\F^\ell$, define the subspace
\begin{equation*}
\mathcal{F}(A_1 \to B_1, \ldots , A_s \to B_s) \coloneqq \{\psi \in \mathcal{L}(\mathbb{F}^\ell, \mathbb{F}^\ell) \; | \; \psi(A_i) \subseteq B_i \; \forall i \in \{1, \ldots , s\}\}
\end{equation*}
for arbitrary subspaces $A_i,B_i$ of $\F^\ell$. Define the value 
\begin{equation*}
\I(A_1 \to B_1, \dots, A_s \to B_s) \coloneqq \dim(\mathcal{F}(A_1 \to B_1, \ldots , A_s \to B_s))
\end{equation*}
When $A_i=B_i$ for each $i$, we adopt the shorthand notation $\mathcal{F}(A_1, \ldots , A_s)$ and $\I(A_1, \ldots , A_s)$ to denote the above quantities. We will also use the mixed notation $\mathcal{F}(A_1, \ldots , A_{s-1}, A_s \to B_s)$ to denote $\mathcal{F}(A_1 \to A_1, \ldots , A_s \to B_s)$ and likewise for $\I(A_1, \ldots , A_{s-1}, A_s \to B_s)$.
\end{defn}

Thus $\I(A_1,\dots,A_s)$ is the dimension of the space of linear maps that map each $A_i$ within itself. We use the notation $\I()$ to suggest such an invariance. The key idea will be to cleverly exploit the invertible maps $\Phi_{i,j}$ associated with each $H_i$ to argue that the dimension $\I(H_1, H_2, \ldots, H_t)$ shrinks by a constant factor whenever we add in an $H_{t+1}$ into the collection. Specifically, we will show that the dimension shrinks at least by a factor of $\frac{r^2-r+1}{r^2}$ for each newly added $H_{t+1}$. Because the identity map is always in $\mathcal{F}(H_1, H_2, \ldots , H_k)$, the dimension $\I(H_1, H_2, \ldots, H_k)$ is at least $1$.  As the ambient space of linear maps from $\F^\ell \to \F^\ell$ has dimension $\ell^2$, this leads to an $O(r \log \ell)$ upper bound on $k$.  We begin with the following lemma.

\begin{lemma}
\label{lem:dim-ub}
Let $U_1, U_2, \ldots , U_s \le \mathbb{F}^p$, $s \ge 2$ be arbitrary subspaces such that $\bigcap_{i = 1}^s{U_i} = \{0\}$. Then following inequality holds:
\begin{equation*}
\sum_{i=1}^s{\dim(U_i)} \le (s-1)\dim\left(U_1 + \ldots + U_s\right) \ .
\end{equation*}
\end{lemma}
\begin{proof}
We proceed by inducting on $s$. Indeed, when $s = 2$, we have from the Principle of Inclusion and Exclusion (PIE)
\begin{equation*}
\dim(U_1) + \dim(U_2) = \dim(U_1 + U_2) + \dim(U_1 \cap U_2) = \dim(U_1 + U_2)
\end{equation*} 
And thus the base case holds. Now, if the inequality holds when $s = p$, then we have via the Principle of Inclusion and Exclusion
\begin{equation}
\sum_{i=1}^{p+1}{\dim(U_i)} = \dim(U_1 + U_2) + \dim(U_1 \cap U_2) + \sum_{i=3}^{p+1}{\dim(U_i)} \label{eq:dim-ub-1}
\end{equation}
By the induction hypothesis, we deduce that Equation~\eqref{eq:dim-ub-1} is at most
\begin{equation}
\dim(U_1 + U_2) + (p-1) \dim((U_1 \cap U_2) + \cdots + U_{p+1}) \label{eq:dim-ub-2}
\end{equation}
And Equation~\eqref{eq:dim-ub-2} is at most
\begin{equation}
p \dim(U_1 + U_2 + \cdots + U_{p+1}) \label{eq:dim-ub-3}
\end{equation}
By combining Equations~\eqref{eq:dim-ub-1}, ~\eqref{eq:dim-ub-2}, and ~\eqref{eq:dim-ub-3}, we deduce that the inequality also holds when $s = p+1$. Since the base case $s = 2$ holds, we therefore conclude that the inequality holds for all integers $s \ge 2$.  
\end{proof}
Next, we prove an identity for MSR subspace families that will come in handy. For the sake of brevity, we use the shorthands $\mathcal{H}_a \coloneqq \{H_1, \ldots , H_a\}$ and $\Phi_{a,0}$ to denote the identity map.

\begin{lemma}
\label{lem:spaces-red}
Given an $(\ell,r)_\F$-MSR subspace family of $H_1, H_2, \ldots , H_k$, we have for any $t \in [k]$ and $i,s \in \{0, 1, \ldots , r-1\}$ that
\small
\begin{equation}
\sum_{j=0}^s{\I(\mathcal{H}_{t-1}, \Phi_{t,i}(H_t) \to \Phi_{t,j}(H_t))} \le s\I(\mathcal{H}_{t-1}, H_t \to 0) + \I(\mathcal{H}_{t-1}, \Phi_{t,i}(H_t) \to \oplus_{j=0}^s\Phi_{t,j}(H_t)) \label{eq:induct-eq}
\end{equation}
\normalsize
\end{lemma}
\begin{proof}
We proceed by inducting on $s$. The base case when $s=0$ is clear as the right hand side simplifies to the left hand side. Now, if Equation~\eqref{eq:induct-eq} holds when $s=p$ and $p < r-1$, then we have via the Principle of Inclusion and Exclusion (PIE) and Equation~\eqref{eq:msr-1}
\begin{equation}
\sum_{j=0}^{p+1}{\I(\mathcal{H}_{t-1}, \Phi_{t,i}(H_t) \to \Phi_{t,j}(H_t))} \label{eq:space-red-1}
\end{equation}
By the induction hypothesis, we deduce that Equation~\eqref{eq:space-red-1} is at most
\begin{equation}
p\I(\mathcal{H}_{t-1}, H_t \to 0) + \I(\mathcal{H}_{t-1}, \Phi_{t,i}(H_t) \to \oplus_{j=0}^p\Phi_{t,j}(H_t)) + \I(\mathcal{H}_{t-1}, \Phi_{t,i}(H_t) \to \Phi_{t,p+1}(H_t)) \label{eq:space-red-2}
\end{equation}
By applying the Principle of Inclusion and Exclusion and Equation~\ref{eq:msr-1}, we deduce that Equation~\eqref{eq:space-red-2} is at most 
\begin{equation}
p\I(\mathcal{H}_{t-1}, H_t \to 0) + \I(\mathcal{H}_{t-1}, \Phi_{t,i}(H_t) \to 0) + \I(\mathcal{H}_{t-1}, \Phi_{t,i}(H_t) \to \oplus_{j=0}^{p+1}\Phi_{t,j}(H_t)) \label{eq:space-red-3}
\end{equation}
And Equation~\eqref{eq:space-red-3} is equal to
\begin{equation}
(p+1)\I(\mathcal{H}_{t-1}, H_t \to 0) + \I(\mathcal{H}_{t-1}, \Phi_{t,i}(H_t) \to \oplus_{j=0}^{p+1}\Phi_{t,j}(H_t)) \label{eq:space-red-4}
\end{equation}
And so combining Equations~\eqref{eq:space-red-1},~\eqref{eq:space-red-2},~\eqref{eq:space-red-3}, and~\eqref{eq:space-red-4}, we deduce that Equation~\eqref{eq:induct-eq} also holds when $s = p+1$. Since the base case $s = 0$ holds, we therefore conclude that the inequality holds for all $s \in \{0, 1, \ldots , r-1\}$. 
\end{proof}

Following Lemma~\ref{lem:spaces-red} and Equation~\eqref{eq:msr-1}, we deduce when $s=r-1$ the following corollary.

\begin{corollary}
\label{cor:dim-red}
Given an $(\ell,r)_\F$-MSR subspace family of $H_1, H_2, \ldots , H_k$, we have for any $t \in [k]$ and $i \in \{0, 1, \ldots , r-1\}$ that
\begin{equation*}
\sum_{j=0}^{r-1}{\I(\mathcal{H}_{t-1}, \Phi_{t,i}(H_t) \to \Phi_{t,j}(H_t))} \le (r-1)\I(\mathcal{H}_{t-1}, H_t \to 0) + \I(\mathcal{H}_{t-1})
\end{equation*}
\end{corollary}

\noindent
We are now ready to establish the key iterative step, showing geometric decay of the dimension $\I(H_1,\dots,H_t)$ in $t$.

\begin{lemma}
\label{lem:geom-decay}
For each $t =1,2, \ldots ,k$, the following holds
\begin{equation}
\I(H_1, \ldots , H_{t-1}, H_t)  \le \left(\frac{r^2-r+1}{r^2}\right) \ \I(H_1, \ldots , H_{t-1})   \ . \label{eq:geom-decay}
\end{equation}
\end{lemma}
\begin{proof}
Recall that by the property of an $(\ell,r)_\F$-MSR subspace family, the maps $\Phi_{t,j}$, $j \in \{0,1,\dots,r-1\}$, leave $H_1,\dots,H_{t-1}$ invariant. Using this it follows that $\I(\mathcal{H}_{t-1}, H_t) = \I(\mathcal{H}_{t-1},$ $\Phi_{t,i}(H_t) \to \Phi_{t,j}(H_t))$ for each $i,j \in \{0,1,\dots,r-1\}$, since we have an isomorphism $\mathcal{F}(\mathcal{H}_{t-1}, H_t) \to \mathcal{F}(\mathcal{H}_{t-1}, \Phi_{t,i}(H_t) \to \Phi_{t,j}(H_t))$ given by $\psi \mapsto \Phi_{t,j} \circ \psi \circ \Phi_{t,i}^{-1}$. Thus we have 
\begin{equation}
r^2 \cdot \I(\mathcal{H}_{t-1}, H_t) = \sum_{i=0}^{r-1}{\sum_{j=0}^{r-1}{\I(\mathcal{H}_{t-1}, \Phi_{t,i}(H_t) \to \Phi_{t,j}(H_t))}} \ . \label{eq:ineq-1}
\end{equation}
Notice the the inner sum is the same as the left hand side in Corollary~\ref{cor:dim-red}. Thus we are able to apply Corollary~\ref{cor:dim-red}  on Equation~\eqref{eq:ineq-1} to find that
\begin{align}
\sum_{i=0}^{r-1}{\sum_{j=0}^{r-1}{\I(\mathcal{H}_{t-1}, \Phi_{t,i}(H_t) \to \Phi_{t,j}(H_t))}} &\le \sum_{i=0}^{r-1}{\left[ (r-1)\I(\mathcal{H}_{t-1}, \Phi_{t,i}(H_t) \to 0) + \I(\mathcal{H}_{t-1}) \right]} \nonumber \\
&= r\I(\mathcal{H}_{t-1}) + (r-1)\sum_{i=0}^{r-1}{\I(\mathcal{H}_{t-1}, \Phi_{t,i}(H_t) \to 0)} \label{eq:ineq-2} \ .
\end{align}
Now we observe that the only linear transformation of $\F^\ell$ that maps $\Phi_{t,i}(H_t) \to 0$ for all $i \in \{0,1,\dots,r-1\}$ simultaneously is the identically $0$ map. This is because $\bigoplus_{j=0}^{r-1} \Phi_{t,j}(H_t) = \F^\ell$ from Equation~\ref{eq:msr-1}. Thus we are in a situation where Lemma~\ref{lem:dim-ub} applies, and we have 
\begin{align}
r\I(\mathcal{H}_{t-1}) + (r-1)\sum_{i=0}^{r-1}{\I(\mathcal{H}_{t-1}, \Phi_{t,i}(H_t) \to 0)} &\le r\I(\mathcal{H}_{t-1}) + (r-1) \cdot (r-1)\I(\mathcal{H}_{t-1}) \nonumber \\
&= (r^2-r+1)\I(\mathcal{H}_{t-1})  \label{eq:ineq-3}
\end{align}
Combining Equations~\eqref{eq:ineq-1}, \eqref{eq:ineq-2}, and \eqref{eq:ineq-3}, we conclude Equation~\eqref{eq:geom-decay} as desired.
\end{proof}

We are now ready to finish off the proof of our claimed upper bound on the size $k$ of an $(\ell,r)_\F$-MSR family.

\begin{proof}[Proof of Theorem~\ref{thm:subspaces-upperbound}]
Since the identity map belongs to the space of $\I(H_1, \ldots , H_k)$, by applying
Lemma~\ref{lem:geom-decay} 
 inductively on $H_1, H_2, \ldots , H_k$, we obtain the inequality
\begin{equation*}
 1 \le \I(H_1, \ldots , H_k) \le \left(\frac{r^2-r+1}{r^2}\right)^k \cdot \ell^2 \ ,
\end{equation*}
from which we find that
\begin{equation*}
k \le \left(\frac{2\ln{\ell}}{\ln\left(\frac{r^2}{r^2-r+1}\right)}\right) \le \left(\frac{2\ln{\ell}}{\frac{r-1}{r^2}}\right) = \left(\frac{2r^2}{r-1}\right)\ln{\ell}
\end{equation*}
where the second inequality follows because $ \ln(1+x) \ge \frac{x}{1+x}$ for all $x > -1$.
We thus have the claimed upper bound.
\end{proof}

\appendix
\section{Proof of Theorem~\ref{theorem:construction}}
In this section, we state and prove an alternate construction of an MSR subspace family of size $(r+1)\log_r{\ell}$. The first construction of an $(\ell, r)_\F$-MSR subspace family of size $(r+1)\log_r{\ell}$ that also satisfied the MDS property was shown in \cite{WTB12} for fields of size more than $k^{r-1}\ell/r$ elements. Without the MDS property, the field size needed to be more than $r$ elements to show that the construction satisfied the node repair property.\medskip

Our construction uses subspaces that are identical to the ones in \cite{WTB12}, but we choose different linear maps that required only two distinct eigenvalues. As a result, our construction works over all fields with more than two elements. It remains a very interesting question whether the MDS property can be additionally incorporated into our construction to yield MSR codes with sub-packetization $r^{k/(r+1)}$ over smaller fields.

\begin{theorem}\label{theorem:construction}
	For $|\F| > 2$ and $r \ge 2$, there exists an $(\ell = r^m,r)_\F$-MSR subspace family of $(r+1)m = (r+1)\log_r(\ell)$ subspaces.
\end{theorem}

In the rest of the section, we will prove the theorem above.

To give a general view of our construction, we first shift our view of the ambient space $\F^\ell = \F^{r^m}$ to $(\F^r)^{\otimes m}$, vectors that consist of $m$ tensored vectors in $\F^r$. We then consider a collection of vectors $T \coloneqq \{v_1, v_2, \ldots , v_r, v_{r+1}\}$, situated in $\F^r$, such that any $r$ of them form a basis in $\F^r$. The subspace $A_{k,i}$ will be all vectors in $(\F^r)^{\otimes m}$ whose $k$'th position in the $m$ tensored vectors is the vector $v_i$.\medskip

The $r-1$ associated linear maps $\Phi_{(k,i), 1}, \ldots ,$ $\Phi_{(k,i),r-1}$ of the subspace $A_{k,i}$ will simply focus on transforming the $k$'th position of each vector while retaining all remaining positions. Specifically, on the $k$'th position, it will scale all vectors in $T \setminus \{v_i\}$. The linear map $\Phi_{(k,i), t}$ will scale $v_{i+t}$ by a factor $\lambda \neq 1$ while all other vectors in $T \setminus \{v_i\}$ will be identically mapped, where the indices are taken modulo $r+1$. That way, everything in $T \setminus \{v_i\}$ will stay almost the same while $v_i$ along with the $r-1$ images of $v_i$ will form a basis for $\F^r$ in the $k$'th position.

\begin{proof}
Let $\ell=r^m$, and let $V = (\F^r)^{\otimes m} \simeq \mathbb{F}^\ell$ be the ambient space. Consider a set of vectors $\{v_1, v_2, \ldots , v_r, v_{r+1}\} \subset \mathbb{F}^r$ for which the first $r$ form a basis in $\mathbb{F}^r$ and satisfy the equation
\begin{equation*}
v_1 + v_2 + \ldots + v_r + v_{r+1} = 0 \label{eq:circularity}
\end{equation*}
For $k \in [m]$ and $i \in [r+1]$, we define our $(r+1)m$ subspaces to be
\begin{equation*}
A_{k, i} \coloneqq \text{span}(v_{i_1} \otimes \ldots \otimes v_{i_m} \; | \; i_j \in [r+1], i_k = i)
\end{equation*}
which is a subspace of $V$. Observe that while the $k$'th position is fixated for any vector in $A_{k,i}$, the remaining  $m-1$ positions are free to choose from any $r$ vectors in $\F^r$. Through this observation, we see that $\dim(A_{k, i}) = r^{m-1} = \ell/r$. \smallskip

To properly define the associated linear maps of the subspace family, it suffices to show their mapping for the basis
\begin{equation*}
S_i \coloneqq \{v_{i_1} \otimes \ldots \otimes v_{i_m} \; | \; i_j \in [r+1] \setminus \{i\}\}
\end{equation*}
of $V$. Since $|\F| > 2$, then we can fix a constant $\lambda \in \F$ with $\lambda \notin \{0, 1\}$, which we will use as an eigenvalue across all $(r-1)(r+1)m$ linear maps. For each $t \in [r-1]$, the linear map $\Phi_{(k, i), t}$ will scale all vectors in $S_i$ whose $k$'th position is $v_{i+t}$ by a factor $\lambda$ and identically all remaining vectors in $S_i$, where indices are taken modulo $r+1$. Namely, for $i_k = i+t$,
\begin{equation*}
v_{i_1} \otimes \ldots \otimes v_{i_k} \otimes \ldots \otimes v_{i_m} \xmapsto{\Phi_{(k, i), t}} v_{i_1} \otimes \ldots \otimes (\lambda v_{i_k}) \otimes \ldots \otimes v_{i_m}
\end{equation*}
And for $i_k \in [r+1] \setminus \{i+t, i\}$,
\begin{equation*}
v_{i_1} \otimes \ldots v_{i_k} \otimes  \ldots \otimes v_{i_m} \xmapsto{\Phi_{(k, i), t}} v_{i_1} \otimes \ldots \otimes v_{i_k} \otimes \ldots \otimes v_{i_m}
\end{equation*}
Observe that all the vectors in the basis $S_i$ are scaled by either $1$ or $\lambda$, which means that the image $\Phi_{(k,i), t}(S_i)$ is also a basis for $V$. This tells us that $\Phi_{(k, i), t}$ is an invertible linear map. It now remains to show Properties~\ref{eq:msr-1} and \ref{eq:msr-2} hold for our given subspaces and linear maps.

To show Property~\ref{eq:msr-1}, we can use Equation~\eqref{eq:circularity} to rewrite $v_i$ as $v_i = -\sum_{j \in [r+1] \setminus \{i\}}{v_i}$. This shows us that when the $k$'th position of a vector is $v_i$, then $\Phi_{(k,i),t}$ will map it as
\begin{equation*}
v_{i_1} \otimes \ldots \otimes v_{i_k} \otimes  \ldots \otimes v_{i_m} \xmapsto{\Phi_{(k, i), t}} v_{i_1} \otimes \ldots \otimes (v_{i} - (\lambda-1)v_{i+t}) \otimes \ldots \otimes v_{i_m}
\end{equation*}
Since $\lambda \neq 1$, then the set $\{v_i, v_i - (\lambda-1)v_{i+1}, \ldots , v_i - (\lambda-1)v_{i+r-1}\}$ forms a basis for $\mathbb{F}^r$. Thus for vector $\mathbf{v} = v_{i_1} \otimes \ldots \otimes v_{i_{k-1}} \otimes v_i \otimes v_{i_{k+1}} \otimes \ldots \otimes v_{i_m}$, the vectors $\{\mathbf{v}, \Phi_{(k,i),1}(\mathbf{v}), \ldots , \Phi_{(k,i),r-1}(\mathbf{v})\}$ span all of $\F^r$ in the $k$'th position. Because we are free to choose any vector in all remaining positions, then are all able to span all of $V$ for all such $\mathbf{v}$. That is, we find that
\begin{equation*}
A_{k, i} \oplus \left(\bigoplus_{t=1}^{r-1}{\Phi_{(k, i), t}(A_{k, i})}\right) = \mathbb{F}^\ell
\end{equation*}
this shows Property~\ref{eq:msr-1}.\medskip

To show \eqref{eq:msr-2}, we start by breaking the subspace $A_{k', i'}$ into two possibilities:
\begin{enumerate}
	
	\item For the case  when $k' \neq k$, the subspace $A_{k', i'}$ remains invariant under each $\Phi_{(k,i),t}$ as they only linearly transform the $k$'th position while retaining all other positions.
	
	\item For the case when $k' = k$ and $i' \neq i$, the subspace $A_{k, i'}$ is an eigenspace for $\Phi_{(k, i), t}$. Namely, when $i' \neq i+t$, $A_{k, i'}$ is the eigenspace of eigenvalue $1$. When $i' = i+t$, the eigenvalue is instead $\lambda$.
	
\end{enumerate}
This shows that \eqref{eq:msr-2} also holds.	
\end{proof}

\section{Proof of the Cutset bound}
\label{app:cutset}
\begin{proof}
Consider an $(n, k, \ell)$-MDS vector code that stores a file $\mathcal{M}$ of size $k\ell$ in storage nodes $s_1, s_2, \ldots, s_n$. The MDS vector code will repair a storage node $s_h$ by making every other storage node $s_i$ communicate $\beta_{i, h}$ bits to $s_h$. From the MDS property, we know that any collection $\mathcal{C} \subseteq [n] \setminus \{h\}$ of $k-1$ of nodes $\{s_i\}_{i \in \mathcal{C}}$ along with $s_h$ is able to construct our original file $\mathcal{M}$. Thus the collective information of these $k$ storage nodes is at least $|\mathcal{M}| =  k\ell$, implying the inequality 
\begin{equation}\label{eq:info-ineq-1}
\sum_{i \in \mathcal{C}}{|s_i|} + \sum_{i \in [n] \backslash \mathcal{C} \cup \{h\}}{\beta_{i, h}} \ge k\ell.
\end{equation}
Since every storage node stores $\ell$ bits ($|s_i| = \ell$), then \eqref{eq:info-ineq-1} reduces down to
\begin{equation}\label{eq:info-ineq-2}
\sum_{i \in [n] \setminus (\mathcal{C} \cup \{h\})}{\beta_{i, h}} \ge \ell.
\end{equation}
Hence \eqref{eq:info-ineq-2} implies that any $n-k$ helper storage nodes collectively communicate at least $\ell$ bits. Thus we find from \eqref{eq:info-ineq-2} by summing over all possible $n-k$ collections of helper storage nodes
\begin{equation}\label{eq:info-ineq-3}
\sum_{i \in [n] \setminus \{h\}}{\beta_{i, h}} \ge \frac{(n-1)}{(n-k)} \cdot \ell.
\end{equation}
Which is the claimed cutset bound. Moreover, to achieve equality for \eqref{eq:info-ineq-3}, equality must be achieved for \eqref{eq:info-ineq-2} over all $n-k$ collections of helper storage nodes. That is possible only when $\beta_{i, h} = \ell/(n-k)$ for all $i \in[n] \setminus \{h\}$. Hence, under optimal repair bandwidth, the total information communicated is $\sum_{i=2}^n{\beta_{i, h}} = (n-1)\ell/(n-k)$ and is only achieved when every helper storage node communicates exactly $\ell/(n-k)$ bits to storage node $s_h$.
\end{proof}


\begin{thebibliography}{10}

\bibitem{AG-stoc19}
O.~Alrabiah and V.~Guruswami.
\newblock An exponential lower bound on the sub-packetization of {MSR} codes.
\newblock In {\em Proceedings of the 51st Annual ACM Syposium on Theory of
  Computing}, pages 979--985, 2019.

\bibitem{BK18}
S.~B. Balaji and P.~V. Kumar.
\newblock A tight lower bound on the sub- packetization level of optimal-access
  {MSR} and {MDS} codes.
\newblock In {\em Proceedings of the {IEEE} International Symposium on
  Information Theory}, pages 2381--2385, 2018.

\bibitem{BBBM94}
M.~Blaum, J.~Brady, J.~Bruck, and J.~Menon.
\newblock {EVENODD:} an optimal scheme for tolerating double disk failures in
  {RAID} architectures.
\newblock In {\em Proceedings of the 21st Annual International Symposium on
  Computer Architecture}, pages 245--254, 1994.

\bibitem{Cadambe_poly}
V.~R. Cadambe, C.~Huang, J.~Li, and S.~Mehrotra.
\newblock Polynomial length {MDS} codes with optimal repair in distributed
  storage.
\newblock In {\em Proc. of Forty Fifth Asilomar Conference on Signals, Systems
  and Computers (ASILOMAR)}, pages 1850--1854, Nov 2011.

\bibitem{CJMRS13}
V.~R. Cadambe, S.~A. Jafar, H.~Maleki, K.~Ramchandran, and C.~Suh.
\newblock Asymptotic interference alignment for optimal repair of {MDS} codes
  in distributed storage.
\newblock {\em IEEE Transactions on Information Theory}, 59(5):2974--2987, May
  2013.

\bibitem{DGWWR10}
A.~G. Dimakis, P.~Godfrey, Y.~Wu, M.~Wainwright, and K.~Ramchandran.
\newblock Network coding for distributed storage systems.
\newblock {\em IEEE Transactions on Information Theory}, 56(9):4539--4551, Sept
  2010.

\bibitem{DRWS-survey}
A.~G. Dimakis, K.~Ramchandran, Y.~Wu, and C.~Suh.
\newblock A survey on network codes for distributed storage.
\newblock {\em Proceedings of the {IEEE}}, 99(3):476--489, 2011.

\bibitem{GHSY12}
P.~Gopalan, C.~Huang, H.~Simitci, and S.~Yekhanin.
\newblock On the locality of codeword symbols.
\newblock {\em {IEEE} Trans. Information Theory}, 58(11):6925--6934, 2012.

\bibitem{GoparajuFV16}
S.~Goparaju, A.~Fazeli, and A.~Vardy.
\newblock Minimum storage regenerating codes for all parameters.
\newblock {\em IEEE Transactions on Information Theory}, 63(10):6318--6328,
  2017.

\bibitem{GTC14}
S.~Goparaju, I.~Tamo, and R.~Calderbank.
\newblock An improved sub-packetization bound for minimum storage regenerating
  codes.
\newblock {\em IEEE Transactions on Information Theory}, 60(5):2770--2779, May
  2014.

\bibitem{GW16}
V.~Guruswami and M.~Wootters.
\newblock Repairing reed-solomon codes.
\newblock {\em IEEE transactions on Information Theory}, 63(9):5684--5698,
  2017.

\bibitem{HPX18}
K.~Huang, U.~Parampalli, and M.~Xian.
\newblock Improved upper bounds on systematic-length for linear minimum storage
  regenerating codes.
\newblock {\em IEEE Transactions on Information Theory}, 65(2):975--984, 2018.

\bibitem{HLKB15}
W.~Huang, M.~Langberg, J.~Kliewer, and J.~Bruck.
\newblock Communication efficient secret sharing.
\newblock {\em IEEE Transactions on Information Theory}, 62(12):7195--7206, Dec
  2016.

\bibitem{PD14}
D.~S. Papailiopoulos and A.~G. Dimakis.
\newblock Locally repairable codes.
\newblock {\em {IEEE} Trans. Information Theory}, 60(10):5843--5855, 2014.

\bibitem{PapDimCad13}
D.~S. Papailiopoulos, A.~G. Dimakis, and V.~Cadambe.
\newblock Repair optimal erasure codes through hadamard designs.
\newblock {\em IEEE Transactions on Information Theory}, 59(5):3021--3037, May
  2013.

\bibitem{RSK11}
K.~V. Rashmi, N.~B. Shah, and P.~V. Kumar.
\newblock Optimal exact-regenerating codes for distributed storage at the {MSR}
  and {MBR} points via a product-matrix construction.
\newblock {\em {IEEE} Trans. Information Theory}, 57(8):5227--5239, 2011.

\bibitem{RTGE18}
A.~S. Rawat, I.~Tamo, V.~Guruswami, and K.~Efremenko.
\newblock {MDS} code constructions with small sub-packetization and
  near-optimal repair bandwidth.
\newblock {\em {IEEE} Trans. Information Theory}, 64(10):6506--6525, 2018.

\bibitem{SAK15}
B.~Sasidharan, G.~K. Agarwal, and P.~V. Kumar.
\newblock A high-rate {MSR} code with polynomial sub-packetization level.
\newblock In {\em Proc. of 2015 IEEE International Symposium on Information
  Theory (ISIT)}, pages 2051--2055, June 2015.

\bibitem{SVK16}
B.~Sasidharan, M.~Vajha, and P.~V. Kumar.
\newblock An explicit, coupled-layer construction of a high-rate msr code with
  low sub-packetization level, small field size and d<(n- 1).
\newblock In {\em 2017 IEEE International Symposium on Information Theory
  (ISIT)}, pages 2048--2052. IEEE, 2017.

\bibitem{TB14}
I.~Tamo and A.~Barg.
\newblock A family of optimal locally recoverable codes.
\newblock {\em {IEEE} Trans. Information Theory}, 60(8):4661--4676, 2014.

\bibitem{zigzag13}
I.~Tamo, Z.~Wang, and J.~Bruck.
\newblock Zigzag codes: {MDS} array codes with optimal rebuilding.
\newblock {\em IEEE Transactions on Information Theory}, 59(3):1597--1616,
  March 2013.

\bibitem{TWB14}
I.~Tamo, Z.~Wang, and J.~Bruck.
\newblock Access versus bandwidth in codes for storage.
\newblock {\em IEEE Transactions on Information Theory}, 60(4):2028--2037,
  April 2014.

\bibitem{TYB-focs17}
I.~Tamo, M.~Ye, and A.~Barg.
\newblock Optimal repair of {R}eed-{S}olomon codes: Achieving the cut-set
  bound.
\newblock In {\em 58th {IEEE} Annual Symposium on Foundations of Computer
  Science}, pages 216--227, 2017.

\bibitem{WTB12}
Z.~Wang, I.~Tamo, and J.~Bruck.
\newblock Long {MDS} codes for optimal repair bandwidth.
\newblock In {\em Proc. of 2012 IEEE International Symposium on Information
  Theory (ISIT)}, pages 1182--1186, July 2012.

\bibitem{YeB17a}
M.~Ye and A.~Barg.
\newblock Explicit constructions of high-rate {MDS} array codes with optimal
  repair bandwidth.
\newblock {\em {IEEE} Trans. Information Theory}, 63(4):2001--2014, 2017.

\bibitem{YeB17b}
M.~Ye and A.~Barg.
\newblock Explicit constructions of optimal-access {MDS} codes with nearly
  optimal sub-packetization.
\newblock {\em {IEEE} Trans. Information Theory}, 63(10):6307--6317, 2017.

\end{thebibliography}

\end{document}